\documentclass[autoref,thm-restate]{lipics-v2021}

\usepackage[utf8]{inputenc}
\usepackage{graphicx} 		

\usepackage{hyperref}
\hypersetup{colorlinks=true,linkcolor=blue,urlcolor=blue,citecolor=red}
\usepackage{todonotes}
\usepackage{amsmath,amsthm,amsfonts,amssymb}

\DeclareMathOperator{\poly}{poly}
\newcommand{\defparaproblem}[4]{
 \vspace{2mm}
\noindent\fbox{
 \begin{minipage}{0.96\textwidth}
 \begin{tabular*}{\textwidth}{@{\extracolsep{\fill}}lr} #1 & \\ \end{tabular*}
 {\textbf{Input:}} #2 \\
 {\textbf{Parameter:}} #3 \\
 {\textbf{Question:}} #4
 \end{minipage}
 }
 \vspace{2mm}
}
\newcommand{\N}{\mathbb{N}}
\newcommand{\M}{\mathcal{M}}
\DeclareMathOperator{\before}{\texttt{before}}
\DeclareMathOperator{\after}{\texttt{after}}
\Copyright{The authors}



\title{On the Complexity of Problems on Tree-structured Graphs}

\author{Hans L. Bodlaender}{Department of Information and Computing Sciences, Utrecht University}{h.l.bodlaender@uu.nl}{ https://orcid.org/
0000-0002-9297-3330}{}

\author{Carla Groenland}{Faculty of Electrical Engineering, Mathematics and Computer Science, Technical University Delft}{c.e.groenland@tudelft.nl}{https://orcid.org/
0000-0002-9878-8750}{This research was done when Carla Groenland was associated with Utrecht University and supported by European grants  CRACKNP (grant agreement No 853234) and GRAPHCOSY (number 101063180).}

\author{Hugo Jacob}{LIRMM, Université de Montpellier, CNRS, Montpellier, France}{hugo.jacob@ens-paris-saclay.fr}{ https://orcid.org/
0000-0003-1350-3240}{}

\author{Marcin Pilipczuk}{University of Warsaw}{malcin@mimuw.edu.pl}{https://orcid.org/
0000-0001-5680-7397}{\flag{logo-erc}\flag{logo-eu}This research is a part of a project that have received funding from the European Research Council (ERC)
under the European Union's Horizon 2020 research and innovation programme
Grant Agreement 714704.}

\author{Micha\l{} Pilipczuk}{University of Warsaw}{michal.pilipczuk@mimuw.edu.pl}{https://orcid.org/
0000-0001-7891-1988}{This research is a part of a project that have received funding from the European Research Council (ERC)
under the European Union's Horizon 2020 research and innovation programme
Grant Agreement 948057.}

\hideLIPIcs

\authorrunning{Bodlaender et al.}
\date{June 2022}

\EventEditors{Holger Dell and Jesper Nederlof}
\EventNoEds{2}
\EventLongTitle{17th International Symposium on Parameterized and Exact Computation (IPEC 2022)}
\EventShortTitle{IPEC 2022}
\EventAcronym{IPEC}
\EventYear{2022}
\EventDate{September 7--9, 2022}
\EventLocation{Potsdam, Germany}
\EventLogo{}
\SeriesVolume{249}
\ArticleNo{13}

\acknowledgements{We would like to thank the organisers of the workshop on Parameterized complexity and discrete optimisation, organised at HIM in Bonn, for providing a productive research environment. We would also like to thank our referees for useful suggestions.}

\relatedversion{An extended abstract of this
paper appeared in the proceedings of IPEC 2022 \cite{BodlaenderGJPP22inIPEC}.}
\begin{document}

\nolinenumbers
\maketitle

\begin{abstract}
    In this paper, we introduce a new class of parameterized problems, which we call XALP: the class of all parameterized
    problems that can be solved  in $f(k)n^{O(1)}$ time and $f(k)\log n$ space on a non-deterministic Turing Machine
    with access to an auxiliary stack (with only top element lookup allowed).
    Various natural problems on `tree-structured graphs' are complete for this class: we show that {\sc List Colouring} and {\sc All-or-Nothing Flow} parameterized by treewidth are XALP-complete. Moreover,  {\sc Independent Set} and {\sc Dominating Set} parameterized by treewidth divided by $\log n$, and {\sc Max Cut} parameterized by cliquewidth are also XALP-complete. 

Besides finding a `natural home' for these problems, we also pave the road for future reductions.
We give a number of equivalent characterisations of the class XALP,
e.g., XALP is the class of problems solvable by an Alternating Turing Machine whose runs have tree size at most $f(k)n^{O(1)}$ and use
    $f(k)\log n$ space.
    Moreover, we introduce  `tree-shaped' variants of {\sc Weighted CNF-Satisfiability} and {\sc Multicolour Clique} that are XALP-complete.
\end{abstract}

\section{Introduction}
A central concept in complexity theory is \emph{completeness} for a
class of problems. Establishing completeness of a problem for a class
pinpoints its difficulty, and gives implications on resources (time, memory or
otherwise) to solve the problem (often, conditionally on complexity
theoretic assumptions). The
introduction of the W-hierarchy by Downey and Fellows in the 1990s
played an essential role in the analysis of the complexity of
parameterized problems~\cite{DowneyF95,DowneyF95a,DowneyF99}. Still, several problems are
suspected not to be complete for a class in the W-hierarchy, and other
classes of parameterized problems with complete problems were introduced,
e.g., the A-, AW-, and M-hierarchies. (See e.g., \cite{AbrahamsonDF95,DowneyF99,FlumG06}.)
In this paper, we introduce a new class of parameterized complexity,
which appears to be the natural home of several `tree structured' parameterized
problems. 
This class, which we call XALP, can be seen as the parameterized
version of a class known in classic complexity theory as 
NAuxPDA[$\poly,\log$] (see \cite{AllenderCLPT14}), or ASPSZ($\log n$, $n^{O(1)}$) \cite{Ruzzo80}.

It can also be seen as the `tree
variant' of the class XNLP, which is the class of parameterized problems that can be solved by a non-deterministic Turing machine using $f(k)\log n$ space in $f(k)n^{O(1)}$ time for some computable function $f$, where $k$ denotes the parameter and $n$ the input size.
It was introduced in 2015 by Elberfeld et al.~\cite{ElberfeldST15}. Recently, several parameterized problems
were shown to be complete for XNLP~\cite{BodlaenderCW22,BodlaenderGNS21,BodlaenderGJJL22}; in this collection,
we find many problems for `path-structured graphs', including well known
problems that are in XP with pathwidth or other linear width measures
as parameter, and linear ordering graph problems like {\sc Bandwidth}. 

Thus, we can view XALP as the `tree' variant of XNLP and as such, we 
expect that many problems known to be in XP (and expected not to be in FPT) when parameterized by treewidth will
be complete for this class.
We will prove the following problems to be XALP-complete in this paper:
\begin{itemize}
    \item {\sc Binary CSP}, {\sc List Colouring} and {\sc All-or-Nothing Flow} parameterized by treewidth;
    \item {\sc Independent Set} and {\sc Dominating Set} parameterized by treewidth divided by $\log n$, where $n$ is the number of vertices of the input graph;
    \item {\sc Max Cut} parameterized by cliquewidth. 
\end{itemize}
The problems listed in this paper should be regarded as examples of a general technique,
and we expect that many other problems parameterized by treewidth, cliquewidth and
similar parameters will be XALP-complete. 
In many cases, a simple modification of an XNLP-hardness proof with
pathwidth as parameter shows XALP-hardness for the same problem with treewidth as parameter. 

In addition to pinpointing the exact
complexity class for these problems, such results have further consequences.
First, XALP-completeness implies XNLP-hardness, and thus hardness for
all classes $\mathrm{W}[t]$, $t \in \mathbb{N}$. Second, a conjecture by
Pilipczuk and Wrochna~\cite{PilipczukW18}, if true, implies that every algorithm for an XALP-complete problem that works in XP time (that is,
 $n^{f(k)}$ time) cannot simultaneously use FPT space (that is,
$f(k)n^{O(1)}$ space). Indeed, typical XP algorithms for problems
on graphs of bounded treewidth use dynamic programming, with tables
that are of size $n^{f(k)}$. 

\subparagraph*{Satisfiability on graphs of small treewidth}
Real-world SAT instances tend to have a special structure to them. One of the measures capturing the structure is the \emph{treewidth} $\mathcal{TW}(\phi)$ of the given formula $\phi$. This is defined by taking the treewidth of an  associated graph, usually a bipartite graph on the variables on one side and the clauses on the other, where there is an edge if the variable appears in the clause.
Alekhnovitch and Razborov \cite{AlekhnovichR11}
raised the question of whether satisfiability of formulas of small treewidth can be checked in polynomial space, which was positively answered by Allender et al.~\cite{AllenderCLPT14}. However, the running time of the algorithm is $3^{\mathcal{TW}(\phi)\log|\phi|}$ rather than $2^{O(\mathcal{TW}(\phi))}|\phi|^{O(1)}$, where $|\phi|=n+m$ for $n$ the number of variables and $m$ the number of clauses. They also conjectured that the $\log|\phi|$ factor in the exponent for the running time cannot be improved upon without using exponential space. 

To support this conjecture, Allender et al.~\cite{AllenderCLPT14} show that {\sc Satisfiability} where the treewidth of the associated graph is
$O(\log n)$ is complete for a class of problems called SAC$^1$: these are the problems that
can be recognised by `uniform' circuits with semi-unbounded fan-in of depth $O(\log n)$ and polynomial size. 
This class has also been shown to be equivalent to classes of problems that are defined using Alternating Turing Machines and non-deterministic Turing machines with access to an auxiliary stack \cite{Ruzzo80,VENKATESWARAN1991380}. 
We define parameterized analogues of the classes defined using Alternating Turning Machines or non-deterministic Turing machines with access to an auxiliary stack, and show these to be equivalent. 
This is how we define our class XALP. 

Allender et al.~\cite{AllenderCLPT14} considers {\sc Satisfiability} where the treewidth of the associated graph is
$O(\log^k n)$ for all $k\geq 1$. 
We restrict ourselves to the case $k=1$ since this is where we could find interesting complete problems, but we expect that a similar generalisation is possible in our setting.

The main contribution of our paper is to transfer definitions and 
results from the classical world to the parameterized setting, 
by which we provide a natural framework to establish the
complexity of many well-known parameterized problems. 
We provide a number of natural XALP-complete problems, 
but we expect that in the future it will be shown that 
XALP is the `right box' for many more problems of interest. 

\subparagraph*{Subsequent work}
Building upon our work, more problems were shown to be hard or complete 
for XALP. The \textsc{Perfect Phylogeny} or 
\textsc{Triangulating Coloured Graphs} problem was shown 
to be XALP-complete by de Vlas~\cite{deVlas23}.
In \cite{BodlaenderGJ22-tpw}, it was shown that 
{\sc Tree-Partition-Width} and {\sc Domino Treewidth} are XALP-complete,
which can be seen as an 
analogue to \textsc{Bandwidth} being XNLP-complete. 
Finding integral 2-commodities was shown to be XALP-complete with treewidth as
parameter in \cite{multicommodityflow}.

In~\cite{BodlaenderGP23}, the complexity class XSLP was introduced; this class
characterises the complexity of several natural problems parameterised
by \emph{treedepth}.

\subparagraph*{Paper overview} 
In Section~\ref{section:definitions}, we give a number of definitions, discuss the classical analogues of
XALP, and formulate a number of key parameterized problems.
Several equivalent characterisations of the class XALP are given in Section~\ref{section:characterisation}.
In Section~\ref{section:complete}, we introduce a `tree variant' of the well-known \textsc{Multicolour Clique} problem. We call this problem {\sc Tree-Chained Multicolour Clique}, and show it to be XALP-hard with 
a direct proof from an acceptance problem of a suitable
type of Turing Machine, inspired by Cook's proof of the NP-completeness
of \textsc{Satisfiability}~\cite{Cook71}. 
In Section~\ref{sec:complete2}, we build on this and give a number of other examples of XALP-complete
problems, including tree variants of \textsc{Weighted Satisfiability} and several 
problems parameterized by treewidth or another tree-structured graph parameter. Some final remarks are made in Section~\ref{section:conclusions}.

\section{Definitions}
\label{section:definitions}
We assume that the reader is familiar with a number of well-known notions from
graph theory and parameterized complexity,
e.g., FPT, the W-hierarchy, clique, independent set, etc. (See e.g.,~\cite{CyganFKLMPPS15}.)

A {\em tree decomposition} of a graph $G=(V,E)$ is a pair $(T=(I,F)$, $\{X_i~|~ i\in T\})$ with
$T=(I,F)$ a tree and $\{X_i~|~i\in I\})$ a family of (not necessarily disjoint)
subsets of $V$ (called {\em bags}) such that $\bigcup_{i\in I} X_i = V$,
for all edges $vw\in E$, there is an $i$ with $v,w\in X_i$, and for all
$v$, the nodes $\{i\in I~|~v\in X_i \}$ form a connected subtree of $T$.
The {\em width} of a tree decomposition $(T, \{X_i~|~ i\in T\})$ is
$\max_{i\in I} |X_i|-1$, and the {\em treewidth} of a graph $G$ 
is the maximum width over all tree decompositions of $G$. A
{\em path decomposition} is a tree decomposition $(T=(I,F)$, $\{X_i~|~ i\in T\})$
with $T$ a path, and the {\em pathwidth} is the minimum width over all path
decompositions of $G$.

\subsection{Turing Machines and Classes}
\label{section:tmclasses}

We assume the reader to be familiar with the basic concept of a Turing Machine (TM).
Here, we
consider TMs that have access to both a fixed input tape (where the machine can only read), and a work tape of specified size (where the machine can both
read and write). We consider \textit{Non-deterministic Turing Machines} (NTM), where
the machine can choose between different transitions, and accepts, if at
least one choice of transitions leads to an accepting state, and 
\textit{Alternating Turing Machines} (ATM), where the machine can both make
non-deterministic steps (accepting when at least one choice leads to
acceptance), and \textit{co-non-deterministic steps} (accepting when both choices
lead to acceptance). We assume a co-non-deterministic step always makes a
binary choice, i.e, there are exactly two transitions that can be done.

Acceptance of an ATM $A$ can be modelled by a rooted binary tree $T$, sometimes  called a {\em{run}} or a {\em{computation tree}} of the machine. Each
node of $T$ is labelled with a configuration of $A$: the 4-tuple consisting of the machine state, work tape contents, location of
work tape pointer, and location of input tape pointer. Each
edge of $T$ is labelled with a transition. The starting configuration
is represented by the root of $T$.
A node with one
child makes a non-deterministic step, and the arc is labelled with a
transition that leads to acceptance; a node with two children makes
a co-non-deterministic step, with the children the configurations after
the co-non-deterministic choice. Each leaf is a configuration with
an accepting state. The {\em time} of the computation is the depth of
the tree; the {\em treesize} is the total number of nodes in this computation tree.
For more information, see e.g.,~\cite{Ruzzo80,PilipczukW18}.
A {\em computation path} is a path from root to leaf in the tree.

We also consider NTMs which additionally have access to an auxiliary stack. For those, a transition can also move the top element of the stack to
the current location of the work tape (`pop'), or put a symbol at the top of the stack (`push'). We stress that only the top element can be accessed or modified, the machine cannot freely read other elements on the stack.

We use the notation N$[t(n,k),s(n,k)]$ to denote languages recognisable by a NTM running in time $t(n,k)$ with $s(n,k)$ working space and A$[t(n,k),s(n,k)]$ to denote languages recognisable by an ATM running in treesize $t(n,k)$ with $s(n,k)$ working space. We note that we are free to put the constraint that \emph{all} runs have treesize at most $t(n,k)$, since we can add a counter that keeps track of the number of remaining steps, and reject when this runs out (similarly to what is done in the proof of Theorem \ref{thm:equiv}).
We write NAuxPDA$[t(n,k),s(n,k)]$ to denote languages recognisable by a NTM with a stack (AUXiliary Push-Down Automaton) running in time $t(n,k)$ with $s(n,k)$ working space. 

Ruzzo \cite{Ruzzo80} showed that for any function $s(n)$, NAuxPDA[$n^{O(1)}$ time,$s(n)$ space] = A[$n^{O(1)}$ treesize, $s(n)$ space]. Allender et al. \cite{AllenderCLPT14} provided natural complete problems when $s(n)=\log^{k}(n)$ for all $k\geq 1$ (via a circuit model called SAC, which we will not use in our paper). Our interest lies in the case $k=1$, where it turns out the parameterized analogue is the natural home of `tree-like' problems.

Another related work by Pilipczuk and Wrochna \cite{PilipczukW18} shows that there is a tight relationship between the complexity of \textsc{3-Colouring} on graphs of treedepth, pathwidth, or treewidth $s(n)$ and problems that can be solved by TMs with adequate resources depending on $s(n)$.

\subsection{From classical to parameterized}
\label{section:classicalanalogues}
In this paper, we introduce the class XALP $ = $ NAuxPDA$[f \text{poly}, f \log]$. Following~\cite{BodlaenderGNS21},
we use the name XNLP for the class N$[f \text{poly}, f \log]$; $f \text{poly}$ is shorthand notation for $f(k)n^{O(1)}$ for
some computable function $f$, and $f \log$ shorthand notation for $f(k)\log n$.

The crucial difference between the existing classical results and our results is that we consider parameterized complexity classes. 
These classes are closed under parameterized reductions, i.e. reductions where the parameter of the reduced instance must be bounded by the parameter of the initial instance. 
In our context, we have an additional technicality due to the relationship between time and space constraints.
While a logspace reduction is also a polynomial time reduction, 
a reduction using $f(k)\log n$ space (XL) could use up to 
$n^{f(k)}$ time (XP). 
XNLP and XALP are closed under \emph{pl-reductions} where the space bound is $f(k) + O(\log n)$ (which implies FPT time), and under \emph{ptl-reductions} running in $f(k)n^{O(1)}$ time \emph{and} $f(k)\log n$ space.

We now give formal definitions.

A {\em parameterized reduction} 
    from a parameterized problem $Q_1 \subseteq \Sigma_1^\ast \times \N$ to a parameterized problem $Q_2 \subseteq \Sigma_2^\ast \times \N$ is a function
    $f \colon \Sigma_1^\ast \times \N \rightarrow \Sigma_2^\ast \times \N$ such that the following holds.
    \begin{enumerate}
        \item For all $(x,k) \in \Sigma_1^{\ast} \times \N$, $(x,k)\in Q_1$ if and only if $f((x,k)) \in Q_2$.
        \item There is a computable function $g$ such that for all $(x,k) \in \Sigma_1^\ast \times \N$, if $f((x,k)) = (y,k')$, then $k' \leq g(k)$.
  \end{enumerate}
  If there is an algorithm that computes $f((x,k))$ in space $O(g(k) + \log n)$, with $g$ a computable function and $n=|x|$ the number of bits to denote $x$, then the reduction is a \emph{parameterized logspace reduction} or \emph{pl-reduction}.
    
  If there is an algorithm that computes $f((x,k))$ in time $g(k)n^{O(1)}$ and space $O(h(k)\log n)$, with $g,h$ computable functions and $n=|x|$ the number of bits to denote $x$, then the reduction is a \emph{parameterized tractable logspace reduction} or \emph{ptl-reduction}.

\section{Equivalent characterisations of XALP}
\label{section:characterisation}
In this section, we give a number of equivalent characterisations of XALP.
\begin{theorem}
\label{thm:equiv}
    The following parameterized complexity classes are all equal.
    \begin{enumerate}
    \item \label{item:stack} NAuxPDA[$f \poly$, $f\log$], the class of parameterized decision
problems for which instances of size $n$ with parameter $k$ can be solved by a non-deterministic Turing machine with $f(k)\log n$ memory in $f(k)n^{O(1)}$ time when given a stack, for some computable function~$f$. 
    \item \label{item:poly_tree} The class of parameterized decision
problems for which instances of size $n$ with parameter~$k$ can be solved by an alternating Turing machine with $f(k)\log n$ memory whose computation tree is a binary tree on $f(k)n^{O(1)}$ nodes, for some computable function~$f$. 
    \item \label{item:embedded-computation} The class of parameterized decision
problems for which instances of size $n$ with parameter~$k$ can be solved by an alternating Turing machine with $f(k)\log n$ memory whose computation tree is obtained from a binary tree of depth $O(\log n)+f(k)$ by subdividing each edge $f(k)n^{O(1)}$ times, for some computable function $f$. 
    \item \label{item:alt} The class of parameterized decision
problems for which instances of size $n$ with parameter $k$ can be solved by an alternating Turing machine with $f(k)\log n$ memory, for which the computation tree has size $f(k)n^{O(1)}$ and uses $O(\log n)+f(k)$ co-non-deterministic steps per computation path,  for some computable function $f$. 
\end{enumerate}
\end{theorem}
\begin{proof}
The proof is similar to the equivalence proofs for the classical analogues, and added for convenience of the reader.
We prove the theorem by proving the series of inclusions \ref{item:stack} $\subseteq $ \ref{item:poly_tree} $\subseteq  $ \ref{item:embedded-computation} $\subseteq  $ \ref{item:alt} $\subseteq  $ \ref{item:stack}.

\noindent\textbf{\ref{item:stack} $\subseteq $ \ref{item:poly_tree}.} Consider a problem that can be solved by a non-deterministic Turing Machine $T$ with a stack and $f(k)\log n$ memory in $f(k)n^{O(1)}$ time. We will simulate $T$ using an alternating Turing machine $T'$.

We place three further assumptions on $T$, which can be implemented by changing the function $f$ slightly if needed.
\begin{itemize}
    \item The Turing machine $T$ has two counters. One keeps track of the height of the stack, and the other keeps track of the number of computation steps. A single computation step may involve several operations; we just need that the running time is polynomially bounded in the number of steps.
\item We assume that $T$ only halts with acceptance when the stack is empty. (Otherwise, do not yet accept,
but pop the stack using the counter that tells the height of the stack, until the stack is empty.)
\item Each pop operation performed by $T$ is a deterministic step. This can be done by adding an extra state to $T$ and splitting a non-deterministic step into a non-deterministic step and a deterministic step if needed.
\end{itemize}
We define a \textit{configuration} as a tuple which includes the state of $T$, the value of the two pointers and the content of the memory. In particular, this does not contain the contents of the stack and so a configuration can be stored using $O(f(k)\log n)$ bits. (Note that the value of both pointers is bounded by $f(k)n^{O(1)}$.)

We will build a subroutine $A(c_1,c_2)$ which works as follows.
\begin{itemize}
    \item The input $c_1$, $c_2$ consists of two configurations with the same stack height.
    \item The output is whether $T$ has an accepting run from $c_1$ to $c_2$ without popping the top element from the stack in $c_1$; the run may pop elements that have yet to get pushed.
\end{itemize}
We write Apply($c$, POP($s$)) for the configuration that is obtained when we perform a pop operation in configuration $c$ and obtain $s$ from the stack. This is only defined if $T$ can do a pop operation in configuration $c$ (e.g. it needs to contain something on the stack).
We define the configuration Apply($c$, PUSH($s$)) in a similar manner, where this time $s$ gets pushed onto the stack.

We let $T'$ simulate $T$ starting from configuration $c_s$ as follows.
Our alternating Turing machine $T'$ will start with the following non-deterministic step: guess $c_a$ the configuration that accepts
at the end of the run. It then performs the subroutine $A(c_s,c_a)$.

We implement $A(c_s,c_a)$ as follows.
A deterministic or non-deterministic step of $T$ is carried out as usual.

If $T$ is in some configuration $c$ and wants to push $s$ to the stack, then let $c'=$ Apply($c$, PUSH($s$)) and let $T'$ perform a non-deterministic step that guesses a configuration $c_2'$ with the same stack height as $c'$ for which the next step is to pop (and the number of remaining computation steps is plausible). Let $c_2 =$ Apply($c_2'$, POP($s$)). We make $T'$ do a co-non-deterministic step consisting of two branches:
\begin{itemize}
    \item $T'$ performs the subroutine $A(c',c_2')$.
    \item $T'$ performs the subroutine $A(c_2',c_a)$.
\end{itemize}
We ensure that in configuration $c_2'$, the number of steps taken is larger than in configuration $c'$. This ensures that $T'$ will terminate.

Since a configuration can be stored using $O(f(k)\log n)$ and $T'$ always stores at most a bounded number of configurations,  $T'$ requires only $O(f(k)\log n)$ bits of memory. 
The computation tree for $T'$ is binary. 
The total number of nodes of the computation tree of $T'$ is $f(k)n^{O(1)}$ since each computation step of $T$ appears at most once in the tree (informally: our co-non-deterministic steps split up the computation path of $T$ into two disjoint parts), and we have added at most a constant number of steps per step of $T$. 
To see this, the computation tree of $T'$ may split a computation path $c\to_{\text{push}} c'\to \dots \to c_2'\to_{\text{pop}} c_2\to\dots \to c_a$ of $T$ into two parts: one branch will simulate $c'\to \dots \to c_2'$ and the other branch will simulate $c_2\to \dots \to c_a$.
At most a constant number of additional nodes (e.g. the node which takes the co-non-deterministic step) are added to facilitate this. Importantly, the configurations implicitly stored a number of remaining computation steps, 
and so $T'$ can calculate from $c',c_2'$ how many steps $T$ is supposed to take to move between $c'$ and $c_2'$.

\noindent \textbf{\ref{item:poly_tree} $\subseteq  $ \ref{item:embedded-computation}.} The intuition behind this proof is to use that any $n$-vertex tree has a tree decomposition of bounded treewidth of depth $O(\log n)$.

Let $A$ be an alternating Turing machine for some parameterized problem with a computation tree of size $f(k)n^{O(1)}$ and $f(k)\log n$ bits of memory. 

We build an alternating Turing machine $B$ that simulates $A$ for which the computation tree is a binary tree which uses $O(f(k) + \log n)$ co-non-deterministic steps per computation branch and $O(f(k)\log n)$ memory. We can after that ensure that there are $f(k)n^{O(1)}$ steps between any two co-non-deterministic steps by adding `idle' steps if needed. 

We ensure that $B$ always has \emph{advice} in memory: 1 configuration for which $A$ accepts. In particular, if $c'$ is the configuration stored as advice when $A$ is in configuration $c$ with a bound of $n$ steps, then $B$ checks if $A$ can get from $c$ to $c'$ within $n$ steps. 

We also maintain a counter for the number of remaining steps: the number of nodes that are left in the computation tree of $A$, when rooted at the current configuration $c$ not counting the node of $c$ itself. In particular, the counter is $0$ if $c$ is supposed to be a leaf.

We let $B$ simulate $A$ as follows. Firstly, if no advice is in memory, it makes a non-deterministic step to guess a configuration as advice.

Suppose that $A$ is in configuration $c$ with $n_0$ steps left. We check the following in order. If $c$ equals the advice, then we accept. If $n_0\leq 0$, then we reject.
If the next step of $A$ is non-deterministic or deterministic step, then we perform the same step. 
The interesting things happen when $A$ is about to perform a co-non-deterministic step starting from $c$ with $n_0$ steps left. If $n_0\leq 1$, then we reject: there is no space for such a step.
Otherwise, we guess $n_1,n_2\geq 0$ such that $n_1+n_2=n_0-2$, and children $c_1,c_2$ of $c$ in the computation tree of $A$. Renumbering if needed, we may assume that the advice $c'$ is supposed to appear in the subtree of $c_1$. We also guess an advice $c'_2$ for $c_2$.
We create a co-non-deterministic step with two branches, one for the computation starting from $c_1$ with $n_1$ steps, and the other from $c_2$ with $n_2$. We describe how we continue the computation starting from $c_1$; the case of $c_2$ is analogous.

Recall that some configuration $c'$ has been stored as advice. We want to ensure that the advice is limited to one configuration. First, we non-deterministically guess a configuration $c''$. We non-deterministically guess whether $c''$ is an ancestor of $c'$. We perform different computation depending on the outcome.
\begin{itemize}
\item Suppose that we guessed that $c''$ is an ancestor of $c'$. We guess integers $\frac13 n_1\leq a,b \leq  \frac23 n_1$ with $a+b=n_1$. We do a co-non-deterministic step: one branch starts in $c_1$ with $c'$ as advice and $a$ steps, the other branch starts in $c'$ with $c''$ as advice and $b$ steps. 
\item Suppose that $c''$ is not an ancestor of $c'$. We guess a configuration $\ell$, corresponding to the least common ancestor of $c'$ and $c''$ in the computation tree. We guess integers $0\leq a,b,a',b' \leq  \frac23 n_1$ with $a+b+a'+b'=n_1$. We perform a co-non-deterministic branch to obtain four subbranches: starting in $c$ with $\ell$ as advice and $a$ steps, $\ell$ with $c'$ as advice and $b$ steps, starting in $\ell$ with $c''$ as advice and $a'$ steps and starting in $c''$ with no advice and $b'$ steps. 
\end{itemize}

In order to turn our computation tree into a binary tree, we may choose to split the single co-non-deterministic step into two steps.

Since at any point, we store at most a constant number of configurations, this can be performed using $O(f(k)\log n)$ bits in memory.

It remains to show that $B$ performs $O(\log n+f(k))$ co-non-deterministic steps per computation path. The computation of $B$ starts with a counter for the number of steps which is at most $f(k)n^{O(1)}$; every time $B$ performs a co-non-deterministic step, this counter is multiplied by a factor of at most $\frac23$. The claim now follows from the fact that $\log(f(k)n^{O(1)})=O(\log n+\log f(k))$.

\noindent \textbf{\ref{item:embedded-computation} $\subseteq  $ \ref{item:alt}.} Let $T$ be an alternating Turing machine using $f(k)\log n$ memory whose computation fits in a tree obtained from a binary tree of depth $d$ by subdividing each edge $f(k)n^{O(1)}$ times. Then $T$ uses $f(k)n^{O(1)}$ time (with possibly a different constant in the $O(1)$-term) and performs at most $d$ co-non-deterministic steps per computation path. Hence this inclusion is immediate.

\noindent \textbf{\ref{item:alt} $\subseteq  $ \ref{item:stack}.}  We may simulate the alternating Turing machine using a non-deterministic Turing machine stack as follows. Each time we wish to do a co-non-deterministic branch, we put the current configuration $c$ onto our stack and continue to the left-child of $c$. Once we have reached an accepting state, we pop an element $c$ of the stack and next continue to the right child of $c$. The total computation time is bounded by the number of nodes in the computation tree and the memory requirement does not increase by more than a constant factor. (Note that in particular, our stack will never contain more than $\log n+f(k)$ elements.)
\end{proof}

Already in the classical setting, it is expected that NL $\subsetneq$ A[poly treesize, log space]. We stress the fact that this would imply XNLP $\subsetneq$ XALP, since we can always ignore the parameter. It was indeed noted in \cite[Corollary 3.13]{AllenderCLPT14} that the assumption NL $\subsetneq$  A[poly treesize, log space] separates the complexity of SAT instances of logarithmic pathwidth from SAT instances of logarithmic treewidth. Allender et al. \cite{AllenderCLPT14} formulate this result in terms of SAC$^1$ instead of the equivalent A[poly treesize, log space]. We expect that a parameterized analogue of SAC can be added to the equivalent characterisation above, but decided to not pursue this here. The definition of such a circuit class requires a notion of `uniformity' that ensures that the circuits have a `small description', which makes it more technical.

\section{XALP-completeness for a tree-chained variant of Multicolour Clique}
\label{section:complete}
\label{section:treechainedmulticolorclique}
Our first XALP-complete problem is a `tree' variant of the well-known \textsc{Multicolour Clique} problem.

\defparaproblem{\textsc{Tree-Chained Multicolour Clique}}{A binary tree $T=(I,F)$, an integer $k$, and for each $i \in I$, a collection of $k$ pairwise disjoint sets of vertices $V_{i,1},\ldots,V_{i,k}$, and a graph $G$ with vertex set $V=\bigcup_{i \in I, j \in [1,k]} V_{i,j}$.}{$k$.}{Is there a set of vertices $W \subseteq V$ such that $W$ contains exactly one vertex from each $V_{i,j}$ ($i \in I, j \in [1,k]$), and for each pair $V_{i,j},V_{i',j'}$ with $i=i'$ or $ii' \in F$, $j,j' \in [1,k]$, $(i,j)\neq(i',j')$, the vertex in $W \cap V_{i,j}$ is adjacent to the vertex in $W \cap V_{i',j'}$?}

This problem is the XALP analogue of the XNLP-complete problem \textsc{Chained Multicolour Clique}, in which the input tree $T$ is a path instead. This change of `path-like' computations to `tree-like' computations is typical when going from XNLP to XALP.

For the \textsc{Tree-Chained Multicolour Independent Set} problem, we have a similar input and question except that we ask for the vertex in $W \cap V_{i,j}$ and the vertex in $W \cap V_{i',j'}$ \emph{not} to be adjacent.
In both cases, we may assume that edges of the graphs are only between vertices of $V_{i,j}$ and $V_{i',j'}$ with $i=i'$ or $ii' \in F$, $j,j' \in [1,k]$, $(i,j) \neq (i',j')$.
We call \emph{tree-chained multicolour clique} (resp. \emph{independent set}) a set of vertices satisfying the respective previous conditions.


Membership of these problems in XNLP seems unlikely, since it is difficult to handle the `branching' of the tree. However, in XALP this is easy to do using the co-non-deterministic steps and indeed the membership follows quickly.
\begin{lemma}
\label{lem:member}
    {\sc Tree-chained Multicolour Clique} is in  XALP.
\end{lemma}
\begin{proof}
We simply traverse the tree $T$ with an alternating Turing machine that uses a co-non-deterministic step when it has to check two subtrees. When at $i\in I$, the machine first guesses a vertex for each $V_{i,j}$, $j \in [k]$. It then checks that these vertices form a multicolour clique with the vertices chosen for the parent of $i$. The vertices chosen for the parent can now be forgotten and the machine moves to checking children of $i$.
The machine works in polynomial treesize, and uses only $O(k\log n)$ space to keep the indices of chosen vertices for up to two nodes of $T$, the current position on $T$.
\end{proof}
We next show that \textsc{Tree-Chained Multicolour Clique} is XALP-hard.
We will use the characterisation of XALP where the computation tree of the alternating Turing machine is a specific tree (\ref{item:embedded-computation}), which allows us to control when co-non-deterministic steps can take place.

Let $\M$ be an alternating Turing machine with computation tree $T=(I,F)$, let $x$ be its input of size $n$, and $k$ be the parameter.
The plan is to encode the configuration of $\M$ at the step corresponding to node $i \in I$ by the choice of the vertices in $V_{i,1},\dots,V_{i,k'}$ (for some $k'=f(k)$). The possible transitions of the Turing Machine are then encoded by edges between $V_i$ and $V_{i'}$ for $ii' \in F$, where $V_j=\bigcup_{\ell \in [1,k']} V_{j,\ell}$.

A configuration of $\M$ contains the same elements as in the proof of Theorem \ref{thm:equiv}:
\begin{itemize}
    \item the current state of $\M$,
    \item the position of the head on the input tape,
    \item the working space which is $f(k)\log n$ bits long, and
    \item the position of the head on the work tape.
\end{itemize}
We partition the working space in $k'=f(k)$ pieces of $\log n$ consecutive bits, and have a set of vertices $V_{i,j}$ for each. 
Formally, we have a vertex $v_{q,p,b,w}$ in $V_{i,j}$ for each tuple $(q,p,b,w)$ where $q$ is the state of the machine, $p$ is the position of the head on the input tape, $b \in \{\texttt{after},\texttt{before}\}\uplus [\log n]$ indicates if the block of the work tape is before or after the head, or its position in the block, and $w$ is the current content of the $j$th block of the work tape.

\begin{figure}
    \centering
    \includegraphics[scale=0.65]{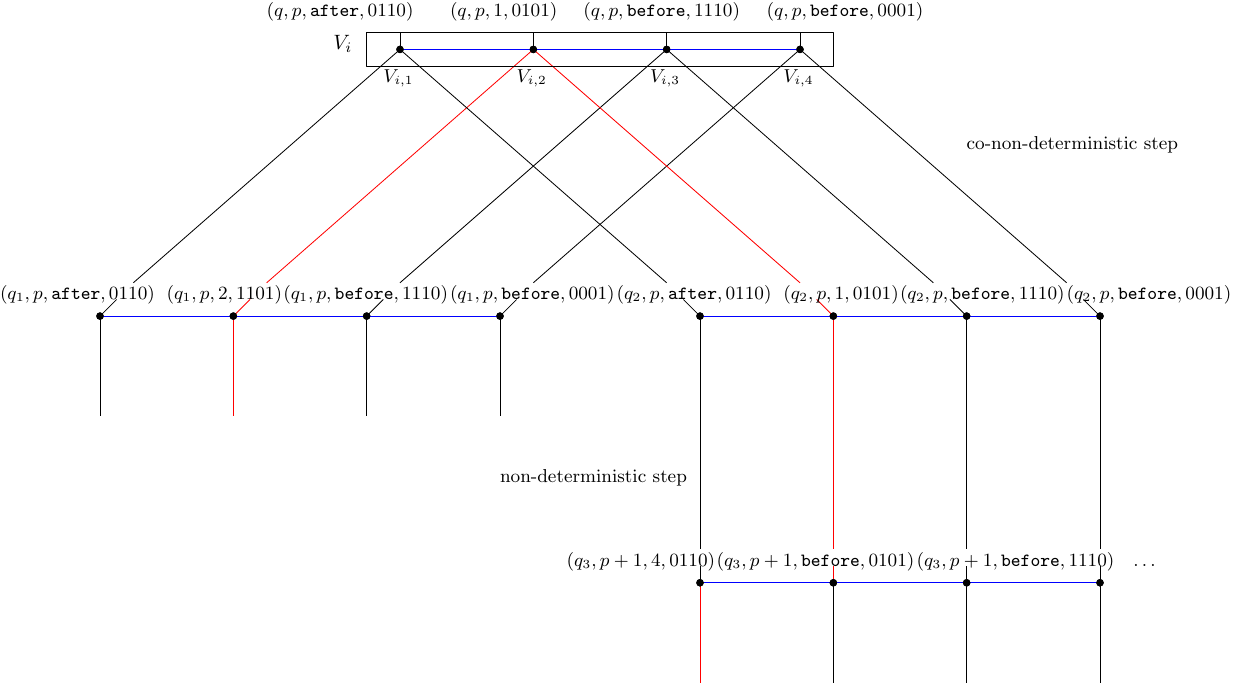}
    \caption{Local structure of a satisfying assignment in the constructed instance of Tree-Chained Multicolor Clique ($k'=4$). The blue edges enforce that the positions of the heads on tapes and the state of the TM are consistent. The black edges enforce that what is written on the work tape does not change in blocks where the head is not present. The red edges enforce that the state, head positions, and the bit at the position of the head on the work tape can be changed exactly by the transitions of the TM. We then add further edges to form cliques, but they do not enforce any constraints.}
    \label{fig:enter-label}
\end{figure}

The edges between vertices of $V_i$ enforce that possible choices of vertices 
correspond to valid configurations. 
There is an edge between $v \in V_{i,j}$ and $w \in V_{i,j+1}$ with corresponding tuples $(q,p,b,w)$ and $(q',p',b',w')$, if and only if $q=q'$, $p=p'$, and either $b'=b \in \{\texttt{after},\texttt{before}\}$, or $b \in [\log n]$ and $b'=\texttt{after}$, or $b=\texttt{before}$ and $b' \in [\log n]$. 

\begin{observation}\label{obs:path-constraint}
    If $v_1,\dots,v_{k'}$ is path with $v_j \in V_{i,j}$, then at most one of the $v_j$ can encode a block with the work tape head, blocks before the head have $b=\before$, blocks after the head have $b=\after$, and all blocks encode the same state and position of the input tape head.
\end{observation}

The edges between vertices of $V_i$ and $V_{i'}$ for $ii' \in F$ enforce that the configurations chosen in $V_i$ and $V_{i'}$ encode configurations with a transition from one to the other.
There is an edge between $v \in V_{i,j}$ and $w \in V_{i',j}$ with corresponding tuples $(q,p,b,w)$ and $(q',p',b',w')$, such that $(b,b') \in \{(\texttt{after},\texttt{after}),(\texttt{before},\texttt{before}),(\texttt{after},1),(\texttt{before},\log n)\}$ if and only if $w=w'$. 
There is an edge between $v \in V_{i,j}$ and $w \in V_{i',j}$ with corresponding tuples $(q,p,b,w)$ and $(q',p',b',w')$, such that $b \in [\log n]$, if and only if, there is a transition of $\M$ from state $q$ to state $q'$ that would write $w'[b]$ when reading $x[p]$ on the input tape and $w[b]$ on the work tape, move the input tape head by $p'-p$ and the work tape by $b'-b$ (where $\texttt{after}=0$ and $\texttt{before}=1+\log n$), and for $\ell \in [\log n]\setminus \{b\}$ $w[\ell]=w'[\ell]$. 

\begin{claim}
    If $v_1,\dots,v_{k'},v'_1,\dots,v'_{k'}$ induce a $2\times k'$ `multicolour grid' (i.e. $v_1,\dots,v_{k'}$ is a path with $v_j \in V_{i,j}$, $v'_1,\dots,v'_{k'}$ is a path with $v'_j \in V_{i',j}$, there are edges $v_jv'_j$ for $j \in [k']$, $ii' \in F$, and $v'_1,\dots,v'_{k'}$ encodes a valid configuration), then $v_1,\dots,v_{k'}$ encodes a valid configuration that can reach the configuration encoded by $v'_1,\dots,v'_{k'}$ using one transition of $\M$.
\end{claim}

\begin{proof}
    This follows easily from the construction but we still detail why this is sufficient when the work tape head moves to a different block.
    
    We consider the case when the head moves to the block before it. That is we consider the case where $v'_j$ encodes $b'=\log n$ and $v_j$ encodes $b=\before$. First, note that there is an edge from $v_jv'_j$ allowing this. We use Observation~\ref{obs:path-constraint} and conclude that $v'_{j+1}$ (if it exists) must encode head position \texttt{after} for its block. The edge $v_{j+1}v'_{j+1}$ then enforces that $v_{j+1}$ encodes head position 1 but it can also exist only if there is a transition of $\M$ that moves the work tape head to the previous block and the written character at the beginning of the block encoded by $v'_{j+1}$ corresponds to such transition. Moving to the next block is a symmetric case.
\end{proof}
We have further constraints on the vertices placed in each $V_{i,j}$ based on what $i$ is in $T$.
\begin{itemize}
    \item If $i$ is in a leaf of $T$, then we only have vertices with a corresponding tuple $(q,p,b,w)$ with $q$ an accepting state.
    \item If $i$ is in a `branching' vertex of $T$ (i.e. $i$ has two children), then we only have vertices with a corresponding tuple $(q,p,b,w)$ with $q$ a universal state.
    \item If $i$ is the root, then only vertices corresponding to the initial configuration are allowed.
    \item Otherwise, we only have vertices with tuples encoding an existential state.
\end{itemize}
Furthermore, we have to make sure that when branching we take care of the two distinct transitions.
We actually assume that $T$ has an order on children for vertices with two children. Then for the edge of $T$ to the first (resp. second) child, we only allow the first (resp. second) transition from the configuration of the parent (which must have a universal state).

We now complete the graph with edges that do not enforce constraints so that we may find a multicolour clique instead of only a $2 \times k'$ multicolour grid.
For every $i \in I$, and $j,j' \in [k']$ such that $|j-j'|>1$, we add all edges between $V_{i,j}$ and $V_{i,j'}$. For every $ii' \in F$, and $j,j' \in [k']$ such that $j \neq j'$, we add all edges between $V_{i,j}$ and $V_{i',j'}$.
It should be clear that to find a multicolour clique for some edge $ii'$ after adding these edges is equivalent to finding a `multicolour grid' before they were added\footnote{Asking for these multicolour grids for each edge of the tree instead of multicolour cliques also leads to an XALP-complete problem but we do not use this problem for further reductions. It could however be used as a starting point for new reductions, the corresponding problem can be expressed as binary CSP on the Cartesian product of $P_{k'}$ and a tree.}.

\begin{claim}
    The constructed graph admits a tree-chained multicolour clique, if and only if, there is an accepting run for $\M$ with input $x$ and computation tree $T$.
\end{claim}

\begin{proof}
    The statement follows from a straight-forward induction on $T$ showing that for each configuration $C$ of $\M$ that can be encoded by the construction at $i \in I$, its encoding can be extended to a tree-chained multicolour clique of the subtree of $T$ rooted at $i$, if and only if there is an accepting run of $\M$ from $C$ with as computation tree the subtree of $T$ rooted at $i$.
\end{proof}
Each $V_{i,j}$ has $O(|Q|n^2\log n)$ vertices (for $Q$ the set of states). Edges are only between $V_{i,j}$ and $V_{i',j'}$ such that $ii' \in F$ or $i=i'$. We conclude that there are $g(k)n^{O(1)}$ vertices and edges in the constructed graph per vertex of $T$, which is itself of size $h(k)n^{O(1)}$ so the constructed instance has size $g(k)h(k)n^{O(1)}$, for $g,h$ computable functions. The construction can even be performed using only $g'(k) + O(\log(n))$ space for some computable function $g'$.
Note also that $k'=f(k)$: the new parameter is bounded by a function of the initial parameter. This shows that our reduction is a parameterized pl-reduction, and we conclude XALP-hardness. Combined with Lemma \ref{lem:member}, we proved the following result.

\begin{theorem}
    {\sc Tree-chained Multicolour Clique} is XALP-complete.
    \label{theorem:tcmc-xalpcomplete}
\end{theorem}
One may easily modify this to the case where each colour class has the same size, by adding isolated vertices.

By taking local complements of the graph, i.e. for each node $i\in V(T)$, we complement the subgraph induced by $\bigcup_{j \in [k']} V_{i,j}$, and for each edge $ii' \in F$, we complement the edge set $E(\bigcup_{j \in [k']} V_{i,j},\bigcup_{j \in [k']} V_{i',j})$, we directly obtain the following result.

\begin{corollary}
{\sc Tree-Chained Multicolour Independent Set} is XALP-complete.
\end{corollary}

\textsc{Multicolour Clique}, \textsc{Chained Multicolour Clique},
and \textsc{Tree-Chained Multi\-colour Clique} can be seen as Binary CSP problems, by replacing vertex choice by assignment choice.

In the \textsc{Binary CSP} problem, we are given a graph $G=(V,E)$,
a set of colours $\mathcal{C}$,
for each vertex $v\in V$ a set of colours $C(v)$, and for each edge $uv\in E$,
a set of pairs of colours $C(u,v) \subseteq \mathcal{C}\times \mathcal{C}$, and ask if we can assign to each vertex $v\in V$
a colour $c(v)\in C(v)$, such that for each edge $uv\in E$, $(f(u),f(v))\in C(u,v)$.

\begin{corollary}
    \textsc{Binary CSP} is XALP-complete with each of the following parameters:
    \begin{enumerate}
        \item treewidth,
        \item treewidth plus degree,
        \item tree-partition width.
    \end{enumerate}
    \label{corollary:bincsp-tw}
\end{corollary}

\begin{proof}
Membership for treewidth as parameter follows as usual. The colour of (uncoloured) vertices is non-deterministically chosen when they are introduced. We maintain the colour of vertices of the current bag in the working space. We use co-non-deterministic steps when the tree decomposition branches. We check that introduced edges satisfy the colour constraint. This uses $O(k\log n)$ space, and runs in polynomial total time. Membership for the two other parameterisations follows from this as well.

Hardness for treewidth plus degree follows from Theorem~\ref{theorem:tcmc-xalpcomplete}.
Suppose we are given an instance of \textsc{Tree-Chained Multicolour Clique}, as described earlier in this section.
We build a graph $H$ by taking for each set $V_{ij}$ a vertex
$v_{ij}$ with $C(v_{ij}) = V_{ij}$, i.e., the vertices in
the \textsc{Tree-Chained Multicolour Clique} now become colours
in the \textsc{Binary CSP} instance. We take an edge $v_{ij}v_{i'j'}$ in $H$ whenever $i=i'$ or $ii'\in F$, $j,j'\in [1,k]$, $(i,j)\neq (i',j')$, and allow for such an edge a pair of
colours $(v,w)$ if and only $vw\in E$. The transformation is mainly
a reinterpretation of a version of \textsc{Multicolour Clique} as
a version of \textsc{Binary CSP}. One easily observes solutions of 
the \textsc{Tree-Chained Multicolour Clique} instance and solutions
of the \textsc{Binary CSP} instance correspond one-to-one to each other, and thus
we have a correct reduction.

Note that $H$ has degree at most $4k-1$ and treewidth at most $2k-1$:
use a tree decomposition $(T=(I,F), \{X_i ~|~i\in T\})$, by
choosing an arbitrary root in $T$, and letting $X_i$ contain
all vertices of the form $v_{ij}$ and $v_{i'j}$ with $i'$ the parent of $i$ in $T$. Hardness for treewidth plus degree as parameter now follows.

Graphs of treewidth $k$ and maximum degree $\Delta$ have
tree-partition width $O(k\Delta)$ (see \cite{Wood09}), and thus
XALP-hardness for tree-partition width as parameter follows.
\end{proof}

\section{More XALP-complete problems}
\label{sec:complete2}
In this section, we prove a collection of problems on graphs, given with a tree-structure, to be complete
for the class XALP. The proofs are of different types: in some cases, the proofs are new, in some cases, 
reformulations of existing proofs from the literature, and in some cases, it suffices to observe that an
existing transformation from the literature keeps the width-parameter at hand bounded.

\subsection{List colouring}
The problems {\sc List Colouring} and {\sc Pre-colouring Extension} with pathwidth as parameter are
XNLP-complete~\cite{BodlaenderGNS21}. We give a simple proof (using a well-known reduction) of XALP-completeness with treewidth
as parameter. Previously, Jansen and Scheffler~\cite{JansenS97} showed that these problem are in XP,
and Fellows et al.~\cite{FellowsFLRSST11} showed $W[1]$-hardness.

\begin{theorem}
    {\sc List Colouring} and {\sc Pre-colouring Extension} are XALP-complete with treewidth as
    parameter.
\end{theorem}

\begin{proof}
Membership follows as the problems are special cases
of \textsc{Binary CSP}.

We first show XALP-hardness of {\sc List colouring}. We reduce from
\textsc{Binary CSP} with treewidth as parameter. 

Suppose we are given a graph $G=(V,E)$, with for each vertex a colour
set $C(v)$, and for each edge a set of allowed colour pairs $C(v,w)\subseteq C(v)\times C(w)$.

First, we can assume that the colour sets $C(v)$ are disjoint.
The hardness proof that gives Corollary~\ref{corollary:bincsp-tw} 
gives such disjoint
sets. (Alternatively, we can rename for each vertex its colours and adjust the
constraints accordingly.)

For each vertex $v\in V$, its list of colours $L(v)=C(v)$.

Now, for each edge $vw\in E$, we remove the edge, but add for each pair
of colours $(c,c')\in C(v)\times C(w)\setminus C(v,w)$ a new vertex $x_{c,c'}$
with $L(x_{c,c'}) = \{ c, c'\}$, and make this new vertex adjacent to $v$ and to $w$.
This new vertex enforces that we cannot use the colour pair $(c,c')$ for the vertices
$v$ and $w$; as we do this for each not allowed colour pair, this ensures that
the restriction of the colouring of $V$ satisfies all colour constraints of the Binary CSP-instance.

The treewidth of the resulting graph is the maximum of the treewidth of $G$ and 2;
take a tree decomposition of $G$, and for each new vertex $x_{c,c'}$ incident to
$v$ and $w$, we take a bag consisting of $\{x_{c,c'}, v, w\}$ and make that bag
incident to a bag that contains $v$ and $w$. This shows the result
for \textsc{List Colouring}.

\smallskip

The standard reduction from {\sc Pre-colouring Extension} to {\sc List colouring} that adds for
each forbidden colour $c$ of a vertex $v$ a new neighbour to $v$ pre-coloured with $c$ does not
increase the treewidth, which shows XALP-hardness for {\sc Pre-colouring Extension} with treewidth as parameter.
\end{proof}

\subsection{Tree variants of Weighted Satisfiability}
From {\sc Tree-Chained Multicolour Independent Set}, we can show XALP-completeness of tree variants of what in \cite{BodlaenderGNS21} was called {\sc Chained Weighted CNF-Satisfiability} and its variants (which in turn are analogues of {\sc Weighted CNF-Satisfiability}, see e.g.~\cite{DowneyF99,DowneyF13}).

\defparaproblem{\textsc{Tree-Chained Weighted CNF-Satisfiability}}{A tree $T=(I,F)$, sets of variables $(X_i)_{i \in I}$, and clauses $C_1,\dots,C_m$, each with either only variables of $X_i$ for some $i \in I$, or only variables of $X_i$ and $X_j$ for some $ij \in F$.}{$k$.}{Is there an assignment of at most $k$ variables in each $X_i$ that satisfies all clauses?}

\defparaproblem{\textsc{Positive Partitioned Tree-Chained Weighted CNF-Satisfiability}}{A tree $T=(I,F)$, sets of variables $(X_i)_{i \in I}$, and clauses of positive literals $C_1,\dots,C_m$, each with either only variables of $X_i$ for some $i \in I$, or only variables of $X_i$ and $X_j$ for some $ij \in F$. Each $X_i$ is partitioned into $X_{i,1},\dots,X_{i,k}$.}{$k$.}{Is there an assignment of exactly one variable in each $X_{i,j}$ that satisfies all clauses?}

\defparaproblem{\textsc{Negative Partitioned Tree-Chained Weighted CNF-Satisfiability}}{A tree $T=(I,F)$, sets of variables $(X_i)_{i \in I}$, and clauses of negative literals $C_1,\dots,C_m$, each with either only variables of $X_i$ for some $i \in I$, or only variables of $X_i$ and $X_j$ for some $ij \in F$. Each $X_i$ is partitioned into $X_{i,1},\dots,X_{i,k}$.}{$k$.}{Is there an assignment of exactly one variable in each $X_{i,j}$ that satisfies all clauses?}

\begin{theorem}
    \textsc{Positive Partitioned Tree-Chained Weighted CNF-Satisfiability}, \textsc{Negative Partitioned Tree-Chained Weighted CNF-Satisfiability}, and \textsc{Tree-Chained Weighted CNF-Satisfiability} are XALP-complete.
\end{theorem}

\begin{proof}
We first show membership for \textsc{Tree-Chained Weighted CNF-Satisfiability}, which implies membership for the more structured versions. We simply follow the tree shape of our instance by branching co-non-deterministically when the tree branches. We keep the indices of the $2k$ variables chosen non-deterministically for the `local' clauses in the working space. We then check that said clauses are satisfied.

We first show hardness for \textsc{Negative Partitioned Tree-Chained Weighted CNF-Satisfiability} by reducing from \textsc{Tree-Chained Multicolour Independent Set}.
For each vertex $v$, we have a Boolean variable $x_v$. We denote by $X_{i,j}$ the set of variables $\{x_v : v \in V_{i,j}\}$, and by $X_i$ the set of variables $\{x_v : v \in V_i\}$. This preserves the partition properties.
For each edge $uv$, we add the clause $\neg x_u \vee \neg x_v$.

\begin{observation}
    $U$ is multicolour independent set if and only if $\{x_u : u \in U\}$ is a satisfying assignment.
\end{observation}

To reduce to \textsc{Positive Partitioned Tree-Chained Weighted CNF-Satisfiability}, we simply replace negative literals $\neg x_v$ for $x_v \in X_{i,j}$ by a disjunction of positive literals $\vee_{y \in X_{i,j}\setminus \{x_v\}} y$. This works because, due to the partition constraint, a variable $x \in X_{i,j}$ is assigned $\bot$ if and only if another variable $y \in X_{i,j} \setminus \{x\}$ is assigned $\top$.

To reduce to \textsc{Tree-Chained Weighted CNF-Satisfiability}, we simply express the partition constraints using clauses. For each $X_{i,j}$, we add the clauses $\vee_{y \in X_{i,j}} y$, and for each pair $\{x,y\} \subseteq X_{i,j}$ the clause $\neg x \vee \neg y$. This enforces that we pick at least one variable, and at most one variable, for each $X_{i,j}$.
\end{proof}

\subsection{Logarithmic Treewidth}

Although XALP-complete problems are in XP and not in FPT, there is a link between XALP and single exponential FPT algorithms on tree decompositions. Indeed, by considering instances with treewidth $k \log n$, where $k$ is the parameter, the single exponential FPT algorithm becomes an XP algorithm. We call this parameter logarithmic treewidth.

\defparaproblem{\textsc{Independent Set} parameterized by logarithmic treewidth}{A graph $G=(V,E)$, with a given tree decomposition of width at most $k \log |V|$, and an integer $W$.}{$k$.}{Is there an independent set of $G$ of size at least $W$?}

\begin{theorem}
    {\sc Independent Set} with logarithmic treewidth as parameter is XALP-complete.
    \label{theorem:islogtw}
\end{theorem}

\begin{proof}
We start with membership which follows from the usual dynamic programming on the tree decomposition.
We maintain for each vertex $v$ in the current bag whether $v$ is in the independent set or not. When introducing a vertex $v$, we non-deterministically decide if $v$ is put in the independent set or not. We reject if an edge is introduced between two vertices of the independent set. We make a co-non-deterministic step whenever the tree decomposition is branching. Since we only need one bit of information per vertex in the bag, this requires only $O(k\log n)$ working space, as for the running time we simply do a traversal of the tree decomposition which is only polynomial treesize.

We show hardness by reducing from \textsc{Positive Partitioned Tree-Chained Weighted CNF-Satisfiability}.
We can simply reuse the construction from \cite{BodlaenderGNS21} and note that the constructed graph has bounded logarithmic treewidth instead of logarithmic pathwidth because we reduced from the tree-chained SAT variant instead of the chained SAT variant. We describe the gadgets for completeness.
First, the SAT instance is slightly adjusted for technical reasons. For each $X_{i,j}$, we add a clause containing exactly its initial variables. This makes sure that the encoding of the chosen variable is valid. We assume the variables in each $X_{i,j}$ to be indexed starting from 0.

\textbf{Variable gadget.} For each $X_{i,j}$, let $t_{i,j}= \lceil \log_2 |X_{i,j}| \rceil$. We add edges $\widehat{0}_\alpha \widehat{1}_\alpha$, $\alpha \in [1,t_{i,j}]$.

\textbf{Clause gadget.} For each clause with $\ell$ literals, we assume $\ell$ to be even by adding a dummy literal if necessary. We add paths $p_0,\dots,p_{\ell + 1}$, and $p'_1,\dots,p'_\ell$.
For $i \in [1,\ell]$, we add the edge $p_ip'_i$. We then add vertex $v_i$ for $i \in [1,\ell]$, which represents the $i$th literal of the clause. Let $b_1\dots b_{t_{i',j'}}$ be the binary representation of the index of the corresponding variable of $X_{i',j'}$. Then $v_i$ is adjacent to $p_i,p'_i$ and the vertices $\widehat{1-b_\alpha}$ for $\alpha \in [1,t_{i,j}]$. For the dummy literal, there is no vertex $v_i$.

The clause gadget has an independent set of size $\ell + 2$ if and only if it contains a vertex $v_i$.
When the variable gadgets have one vertex in the independent set on each edge, a vertex $v_i$ of a clause can be added to the independent set only if the independent set contains exactly the vertices of the variable gadget that give the binary representation of the variable corresponding to $v_i$.

Hence, the SAT instance is satisfiable if and only if there is an independent set of size $\sum_{i,j} t_{i,j} + \sum_i 2 + \ell_i$ in our construction.
\end{proof}

\begin{corollary}
The following problems are XALP-complete with logarithmic treewidth as parameter: \textsc{Vertex Cover},
\textsc{Red-Blue Dominating Set}, \textsc{Dominating Set}.
\end{corollary}

\begin{proof}
The result for \textsc{Vertex Cover} follows directly from Theorem~\ref{theorem:islogtw} and the well known fact that a graph
with $n$ vertices has a vertex cover
of size at most $L$, iff it has an independent set of size at least $n-L$.
Viewing \textsc{Vertex Cover} as a special case of \textsc{Red-Blue Dominating Set} gives the following graph: subdivide all edges of $G$, and ask if a set of $K$ original (blue) vertices dominates all new (red) 
subdivision vertices; as the subdivision step does not increase the treewidth,
XALP-hardness of \textsc{Red-Blue Dominating Set} with treewidth as parameter follows. 
To obtain XALP-hardness of \textsc{Dominating Set}, add to the instance $G'$ of \textsc{Red-Blue Dominating Set},
two new vertices $x_0$ and $x_1$ and
edges from $x_1$ to $x_0$ and all blue vertices; the treewidth increases by at most one, and the minimum size of a dominating set 
in the new graph is exactly one larger than the minimum size of a red-blue dominating set in $G'$.
Membership in XALP is shown similarly to the proof of Theorem~\ref{theorem:islogtw}.
\end{proof}

\subsection{All-or-Nothing Flow}
A \emph{flow network} is a tuple $(G=(V,E),s,t,c)$ with $G=(V,E)$ a
directed graph, $s\in V$ a vertex called the \emph{source},
$t\in V$ a vertex called the \emph{sink}, and $c: E\rightarrow \mathbb{Z}^+$ a \emph{capacity function}, assigning to each arc a
positive capacity, given in unary.
A \emph{flow} in flow network $(G=(V,E),s,t,c)$ is a function
$f: E \rightarrow \mathbb{N}$, that assigns to each arc a non-negative flow value, such that
\begin{enumerate}
    \item for each $e\in E$: $0\leq f(e)\leq c(e)$, and
    \item for each $v\in V\setminus \{s,t\}$: $\sum_{xv\in E} f(xv) = \sum_{vy\in E} f(vy)$.
\end{enumerate}
The \emph{value} of a flow $f$ is $\sum_{sy\in E} f(sy) - \sum_{xs\in E} f(xs)$. 
For more background on flow, we refer to the various text books on algorithms or flow, e.g.,~\cite{AhujaMO93}.

A flow $f$ is an \emph{all-or-nothing flow} if for each
$e\in E$: $f(e)\in \{0,c(e)\}$, i.e., when there is flow over an 
arc then all capacity of the arc is used. Deciding whether
there is an all-or-nothing flow of a given value $r$ in a given flow network is NP-complete~\cite{Alexandersson01}. In~\cite{BodlaenderCW22}, it was shown that this problem is XNLP-complete with pathwidth as parameter. As that proof uses a reduction
from a problem that has no `tree variant' yet, we use here a different proof.

\defparaproblem{\textsc{All-or-Nothing Flow} parameterized by treewidth}{A flow network $(G=(V,E),s,t,c)$, with a given tree decomposition of width at most $k$, and an integer $r$.}{$k$.}{Is there an all-or-nothing flow from $s$ to $t$ in $G$ with value exactly $r$?}

We remark that the proof below can also be used (without changes)
to show XALP-completeness for the variant where we ask whether there
is a flow of value \emph{at least} $r$.

\begin{theorem}
    \textsc{All-or-Nothing Flow} parameterized by treewidth is
    XALP-complete.
\end{theorem}

\begin{proof}
    Membership can be shown in the usual way. For each introduced arc $e$, we guess whether it is used ($f(e)=c(e)$) or not ($f(e)=0$), and the status of a bag is a function that gives
    for each vertex $v$ in the bag the difference of the total inflow so far and the total outflow so far ($\sum_{vy\in E} f(vy) - \sum_{xv\in E} f(xv)$). Because the capacities are given in unary, storing these values requires only $O(k\log(n))$ bits.

    For the hardness proof, we reduce from \textsc{BinaryCSP} with
    treewidth plus degree as parameter.
    First, we build an
    equivalent instance where all sets of colours $C(v)$ are
    disjoint: $v\neq w \Rightarrow C(v)\cap C(w)=\emptyset$. This
    can be easily done by simple adaption of the instance.

So, we assume we are given a graph $G=(V,E)$ of treewidth at most $k$ and degree at most $d$, a set of colours
$\mathcal{C}$, for each $v\in V$ a set $C(v)\subseteq \mathcal{C}$ with these sets disjoint, and for each ordered pair of vertices
$v,w$ that forms an edge, a set of allowed colour pairs
$C(v,w) \subseteq C(v) \times C(w)$; and finally, we have 
a tree decomposition of $G$ of width at most $k$.

In the proof, we use the technique of representing colours by
flow values in a Sidon set; a similar technique was used in
\cite{multicommodityflow}.

A \emph{Sidon set} is a set of positive integers $\{x_1, x_2, \ldots, x_n\}$ such that each pair of integers from the set has
a different sum, i.e., for $\{i,i'\}\neq\{j,j'\}$, $x_i+x_{i'} \neq x_j + x_{j'}$. Sidon sets are also known as Golomb rulers.
Erd\H{o}s and Tur\'{a}n~\cite{ErdosT41} gave a method to construct
Sidon sets; as discussed in \cite{BodlaenderW20}, their construction
implies the following.

\begin{theorem} \label{thm:sidon}
A Sidon set with $n$ elements in $\{1,2,\ldots, 4n^2\}$ can be found in $O(n \sqrt{n})$ time and logarithmic space.
\end{theorem}

The next step in the construction is to build a Sidon set
with $|\mathcal{C}|$ elements, following the construction
of Erd\"{o}s and Tur\'{a}n~\cite{ErdosT41}. Write $L=4n^2+1$.
Note that if we take a Sidon set, and add the same number to
each element of the set, we again obtain a Sidon set. Now, we
add $2L$ to each element of the just created Sidon set. Each of
these numbers is between $2L+1$ and $3L-1$; we assign to each
color $\gamma\in \mathcal{C}$ a unique element $S(\gamma)$ from this latter set. I.e., for each $\gamma\in C$, we have $2L < S(\gamma)< 3L$, and
different pairs of colours have a different sum of their $S(\gamma)$ values.

In the flow network we are constructing, each vertex $v \in V$ from $G$ is represented
by $d_v+1$ vertices, with $d_v$ the degree of $v$ in $G$.
Call these vertices $v_0, v_1, v_2, \ldots, v_{d_v}$.
The construction relies on two gadgets: one that models assigning a colour to a vertex, and one that models checking for an edge that the
assigned colours are an allowed pair.

In the description, we allow first parallel arcs with different
capacities. As a final step, we will subdivide each arc once --- if we subdivide an arc with some capacity $\alpha$, then both
resulting arcs get capacity $\alpha$ as well. Clearly, the network without subdivisions has an all-or-nothing flow with the required value, if and only if the network with subdivisions has such. Also,
given a tree decomposition of the network with parallel arcs, we
can build one of the same width 
(assuming the width is at least 2) for the network with subdivisions, 
as follows: if we subdivide an arc $vw$ to $vx$ and $xw$, then we
add a new bag containing $v$, $w$, and $x$ and make that bag incident to a bag that contains $v$ and $w$ --- the latter exists due to the definition of a tree decomposition.

To model the assignment of a colour to a vertex, we have a gadget
with one addition vertex $x_v$. We have an arc from $s$ to $x_v$
with capacity $6L$; for each colour $\gamma\in C(v)$, we have
an arc from $x_v$ to $v_0$ with capacity $S(\gamma)$ and an arc
from $x_v$ to $t$ with capacity $6L-S(\gamma)$.
The intuition is as follows: setting to colour of $v$ to $\gamma$
corresponds to sending $6L$ from $s$ to $x_v$, $S(\gamma)$ from
$x_v$ to $v_0$ and $6L-S(\gamma)$ from $x_v$ to $t$.
See Figure~\ref{fig:a-o-n1}.

\begin{figure}
    \centering
    \includegraphics{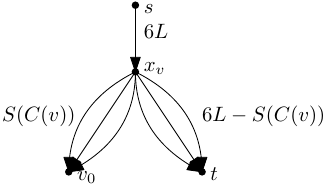}
    \caption{Gadget that chooses a colour for a vertex.  For each $\gamma \in C(v)$, there is an arc to $v_0$ with capacity $S(\gamma)$ and  for each $\gamma \in C(v)$, there is an arc to $t$ with capacity $6L-S(\gamma)$. (The shorthand notations $S(C(v))$ and $6L-S(C(v))$ indicate these values respectively.) }
    \label{fig:a-o-n1}
\end{figure}

Next, we describe the gadget that models a check that the pair of
colours assigned to the endpoints of an edge $\{v,w\}$ is
in $C(v,w)$. 

We assume we have an ordering of the vertices; for each
vertex, order its neighbours accordingly.
Suppose $w$ is the $j$th neighbour of $v$,
and $v$ is the $j'$th neighbour of $w$; thus, $1\leq j\leq d_v$,
and $1\leq j' \leq d_w$. 
The gadget has two additional vertices: $y_{vw}$ and $z_{vw}$.
(We have one gadget per edge rather than per arc, and so misuse notation a little:
$y_{vw}$ and $y_{wv}$ represent the same vertex, likewise for 
$z_{vw}$ and $z_{wv}$.) 
For each $\gamma\in C(v)$, we have an arc from $v_{j-1}$ to $y_{vw}$
with capacity $S(\gamma)$, and an arc from $z_{vw}$ to $v_j$ with
capacity $S(\gamma)$.
For each $\gamma'\in C(w)$, we have an arc from $w_{j'-1}$ to $y_{vw}$
with capacity $S(\gamma')$, and an arc from $z_{vw}$ to $w_{j'}$ with
capacity $S(\gamma')$.
For each $(\gamma,\gamma')\in C(v,w)$, we have an arc
from $y_{vw}$ to $z_{vw}$ with capacity $S(\gamma)+S(\gamma')$.
See Figure~\ref{fig:a-o-n2}.

\begin{figure}
    \centering
    \includegraphics{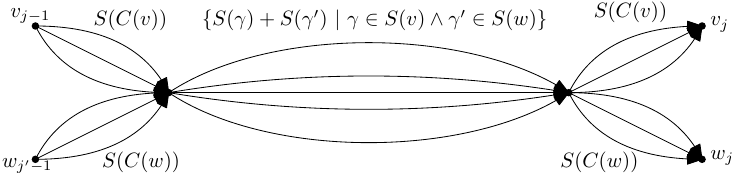}
    \caption{Gadget that checks if an edge is properly coloured. Shorthand notation: $S(C(v))$: for each $\gamma \in C(v)$, there is an arc with capacity $S(\gamma)$. Similar for $S(C(w))$.}
    \label{fig:a-o-n2}
\end{figure}

The intuition here is as follows. If $v$ has colour $\gamma$
and $w$ colour $\gamma'$, then we send $S(\gamma)$ from $v_{j-1}$ to
$y_{vw}$, $S(\gamma')$ from $w_{j'-1}$ to $y_{vw}$,
$S(\gamma)+S(\gamma')$ from $y_{vw}$ to $z_{vw}$, $S(\gamma)$ from $z_{vw}$ to
$v_{j}$ and $S(\gamma')$ from $z_{vw}$ to $w_{j'}$.  The property of
Sidon sets ensures that we cannot reroute flow in another way, i.e.,
the  amount of flow that departs from $v_{j-1}$ equals the amount of flow that arrives
at $v_j$.

Finally, we have for each vertex $v\in V$, and $\gamma\in C(v)$,
an arc from $v_{d_v}$ to $t$ with capacity $S(\gamma)$.

Let $H$ be the resulting graph, with $s$ the source, $t$ the sink.

\smallskip

We first will show how to build a tree decomposition of width
$O(dk)$ of $H$. Take a tree decomposition $(T=(I,F), \{X_i ~|~i\in I\})$ of $G$. For each $i\in I$, we take a bag $Y_i$ that contains $s$; $t$; for each vertex $v\in X_i$, the vertices $x_v$,
$v_{0}, v_1, \ldots, v_{d_v}$; and for each edge $\{v,w\}$ in $G$,
if $v,w\in X_i$, the vertices $y_{vw}$ and $z_{vw}$. It is easy
to see that this indeed is a tree decomposition of $H$ (still with parallel arcs), and each
bag clearly is of size $O(dk)$. As discussed above, we can obtain
an equivalent instance without parallel arcs by subdividing arcs
and adding bags with three vertices.

\begin{lemma}
Suppose $G$ has $n$ vertices.
    There is an all-or-nothing flow with value $6Ln$ from $s$ to
    $t$ in $H$, if and only $G$ has a colouring $h$ with for
    each $v\in V$, $h(v)\in C(v)$, and for each $\{v,w\}\in E$, $(h(v),h(w))\in C(v,w)$.
\end{lemma}

\begin{proof}
    First, suppose that $G$ has a colouring that fulfils the demands.
Build a flow as follows. For each vertex $v\in V$ coloured with $\gamma$, send $6L$ from $s$ to $x_v$, $S(\gamma)$ from
$x_v$ to $v_0$, $6L-S(\gamma)$ from $x_v$ to $t$,
$S(\gamma)$ from each $v_i$ to the corresponding $y_{vw}$, $S(\gamma)$ from
each $z_{vw}$ to the corresponding $v_i$, and $S(\gamma)$ from
$v_{d_v}$ to $t$. 
For each edge $\{v,w\}$, if $v$ is coloured $\gamma$ and $w$ is coloured
$\gamma'$, then send $S(\gamma)+S(\gamma')$ flow from $y_{vw}$ to $z_{vw}$.
One can check that this is an all-or-nothing flow from $s$ to $t$; its value is $6Ln$, as $s$ sends $6L$ flow to each vertex $x_v$.

\smallskip

Now, suppose we have a flow $f$ with value $6Ln$ from $s$ to $t$.
As the total capacity of all outgoing arcs from $s$ equals $6Ln$,
each of these arcs is used, so each $x_v$ receives $6L$ inflow.
Outgoing arcs from $x_v$ have capacities of the form $S(\gamma)$ or
$6L-S(\gamma)$, but as for each colour $\gamma$, $2L<S(\gamma)<3L$, the only possible way to have an outflow of exactly $6L$ is to send
$S(\gamma)$ flow over one outgoing arc, and $6L-S(\gamma)$ flow over another outgoing arc, 
for some colour $\gamma$.

Thus, each vertex $v_0$ receives flow $S(\gamma)$ over one of its
incoming arcs, for some colour $\gamma$. By construction $\gamma\in C(v)$.
Let $h$ be the colouring of $G$ obtained by colouring each $v$
with the colour $\gamma$ such that the flow from $x_v$ to $v_0$ equals
$S(\gamma)$. We claim that each vertex $v_i$ receives $S(h(v))$ flow,
and that we send for each edge $\{v,w\}$ in $G$, $S(h(v))+S(h(w))$ flow from $y_{vw}$ to $z_{vw}$.
We show this by induction. Note that the claim holds for $v_0$ for all $v\in V$.
Observe that $H$ is acyclic.
Suppose $w$ is the $j$th neighbour of $v$ and $v$ is the $j'$th neighbour of $w$.
By the induction hypothesis, $v_{j-1}$ receives $S(h(v))$ flow, and $w_{j'-1}$
receives $S(h(w))$ flow. So, $y_{vw}$ receives   $S(h(v))+S(h(w))$
flow which it sends to $z_{vw}$. Now, we use the Sidon property: the only
possible way for $z_{vw}$ to send out this flow is to
send $S(h(v))$ flow to $v_j$ and $S(h(w))$ flow to $w_{j'}$; any
other combination of flows would imply a second pair of values
in the Sidon set with the same sum. This shows that the induction
hypothesis holds.

Now, we use the Sidon property for the second time. As we send $S(h(v))+S(h(w))$ flow from $y_{vw}$ to $z_{vw}$, there must
be an arc between these vertices with this capacity. So, 
there is a pair $(\gamma'',\gamma''')\in C(v,w)$ with $S(\gamma'')+S(\gamma''')= S(h(v))+S(h(w))$. By the Sidon property, $\{\gamma'',\gamma'''\} = \{h(v),h(w)\}$, and as the sets $S(v)$ for the vertices $v$ are
disjoint, we have $\gamma''=h(v)$ and $\gamma'''= h(w)$, so $(h(v),h(w)\in C(v,w)$. As this holds for each edge, we have a colouring that
satisfies the constraints.
\end{proof}
By observing that the transformation can be done in logarithmic space, the result now follows.
\end{proof}

\subsection{Other problems}
\label{section:otherproblems}
Several XALP-hardness proofs follow from known reductions. Membership is usually easy to prove, by observing that the known
XP-algorithms can be turned into XALP-membership by guessing table entries, and using the stack to store the information for a
left child when processing a right subtree.

\begin{corollary}
The following problems are XALP-complete:
\begin{enumerate}
    \item {\sc Chosen Maximum Outdegree}, {\sc Circulating Orientation}, {\sc Minimum Maximum Outdegree}, {\sc Outdegree Restricted Orientation}, and
    {\sc Undirected Flow with Lower Bounds}, with the treewidth as parameter.
    \item \textsc{Max Cut} and \textsc{Maximum Regular Induced Subgraph} with clique-width as parameter.
\end{enumerate}
\end{corollary}

\begin{proof}
(1): The reductions given in \cite{BodlaenderCW22} and \cite{Szeider11} can be used; one easily observes that these reductions keep the treewidth of the constructed instance bounded by a function of the treewidth of the original instance (often, a small additive constant is added.)

(2): The reductions given in \cite{BodlaenderGJJL22} can be reused with minimal changes, only the bound on linear clique-width becomes a bound on clique-width because of the `tree-shape' of the instance to reduce. 
\end{proof}

{\sc Chosen Maximum Outdegree}, {\sc Circulating Orientation}, {\sc Minimum Maximum Outdegree}, 
{\sc Outdegree Restricted Orientation}, and {\sc Undirected Flow with Lower Bounds},
together with {\sc All-or-Nothing Flow} were shown to be XNLP-complete with pathwidth as
parameter in \cite{BodlaenderCW22}. Gima et al.~\cite{GimaHKKO22} showed that {\sc Minimum Maximum
Outdegree} with vertex cover as parameter is $W[1]$-hard. For related results, see also \cite{Szeider11}.

\section{Conclusions}
\label{section:conclusions}
We expect many (but not all) problems that are (W[1]-)hard and in XP for treewidth as parameter to be XALP-complete;
our paper gives good starting points for such proofs. Let us give an explicit example. The \textsc{Pebble Game Problem} \cite{DowneyF13,KasaiAI79}
parameterized by the number of pebbles is complete for XP, which is equal to XAL=A[$\infty,f\log$]. The problem corresponds to deciding whether there is a winning strategy in an adversarial two-player game with $k$ pebbles on a graph where the possible moves depend on the positions of all pebbles. We can expect variants with at most $f(k) + O(\log n)$ moves to be complete for XALP.

Completeness proofs give a relatively precise complexity classification of problems. In particular, XALP-hardness proofs indicate that we do not expect a deterministic algorithm to use less than XP space if it runs in XP time. Indeed the inclusion of XNLP in XALP is believed to be strict, and already for XNLP-hard problems we have the following conjecture.

\begin{conjecture}[Slice-wise Polynomial Space Conjecture~\cite{PilipczukW18}]
\label{conj}
No XNLP-hard problem has an algorithm that runs in $n^{f(k)}$ time and $f(k)n^{O(1)}$ space, with $f$ a computable function, $k$ the parameter, $n$ the input size.
\end{conjecture}

While XNLP and XALP give a relatively simple framework to classify problems in terms of simultaneous bound on space and time, the parameter is allowed to blow up along the reduction chain. One may want to mimic the fine grained time complexity results based on the (Strong) Exponential Time Hypothesis. In this direction, one could assume that Savitch's theorem is optimal as was done in \cite{ChenEM19}.

Since XNLP is above the W-hierarchy, it could be interesting to study the relationship of XALP with some other hierarchies like the A-hierarchy and the AW-hierarchy. 
It is also unclear where to place \textsc{List-colouring} parameterized by tree-partition-width\footnote{A tree-partition of a graph $G$ is a partition of $V(G)$ into (disjoint) bags $(B_i)_{i\in V(T)}$, where $T$ is a tree,  such that $uv\in V(G)$ implies that the bags of $u$ and $v$ are the same or adjacent in $T$. 
The width is the size of the largest bag, and the tree-partition-width of $G$ is found by taking the minimum width over all tree-partitions of $G$.}.
It was shown to be in XL and W[1]-hard \cite{ListColTrees} but neither look like good candidates for completeness.

\bibliographystyle{plainurl}
\bibliography{references}

\end{document}